\documentclass{llncs}

\usepackage[ruled]{algorithm2e}

\SetAlFnt{\small}
\SetAlCapFnt{\small}
\SetAlCapNameFnt{\small}
\SetAlCapHSkip{0pt}
\IncMargin{-\parindent}

\usepackage{amsmath}
\usepackage{amsfonts}

\usepackage{graphicx}
\usepackage{float}
\usepackage[lofdepth, lotdepth]{subfig}

\newcommand{\prob}{{\bf \mbox{\bf Pr}}}

\begin{document}

\title{Triadic Consensus}
\subtitle{A Randomized Algorithm for Voting in a Crowd}
\author{Ashish Goel\inst{1} \and David Lee\inst{1}}
\institute{Stanford University \email\{ashishg, davidtlee\}@stanford.edu}

\maketitle

\begin{abstract}
    Typical voting rules do not work well in settings with many candidates. If there are just several hundred candidates, then even a simple task such as choosing a top candidate becomes impractical. Motivated by the hope of developing group consensus mechanisms over the internet, where the numbers of candidates could easily number in the thousands, we study an urn-based voting rule where each participant acts as a voter and a candidate. We prove that when participants lie in a one-dimensional space, this voting protocol finds a $(1-\epsilon/\sqrt{n})$ approximation of the Condorcet winner with high probability while only requiring an expected $O(\frac{1}{\epsilon^2}\log^2 \frac{n}{\epsilon^2})$ comparisons on average per voter. Moreover, this voting protocol is shown to have a quasi-truthful Nash equilibrium: namely, a Nash equilibrium exists which may not be truthful, but produces a winner with the same probability distribution as that of the truthful strategy. 
\end{abstract}

\section{Introduction}

Voting is often used as a method for achieving consensus among a group of individuals. This may happen, for example, when a committee chooses a representative or friends go out to watch a movie. When the group is small, this process is relatively easy; however, for larger groups, the typical requirement of ranking all candidates becomes impractical and heuristics are often applied to narrow down opinions to a few representative ones before a vote is taken.

This problem of large-scale preference aggregation is even more interesting in light of the rising potential of crowdsourcing. Suppose that a city government wanted to ask its constituencies to contribute solutions for an ``ideal budget that cuts 50 percent of the deficit''.\footnote{See, for example, widescope.stanford.edu, aimed at tackling the federal budget deficit} Soliciting such proposals may be relatively straightforward; however, it is not clear how these proposals should be aggregated. In particular, a participant cannot even look through each proposal, making seemingly simple tasks such as choosing top ranked proposals, difficult. A solution to this problem would enable a new level of collaboration, a key step towards unleashing the full potential of crowdsourcing.

In this paper, we propose a randomized voting rule designed for scenarios like the above. In our problem setting, each participant submits exactly one proposal, representing his or her stance on the question of interest. A random triad of participants is then selected and each selected member is made to vote between the other two. Roughly speaking (details are elaborated in Sect. \ref{subsec:TriadicConsensus}), if there is a three-way tie, the participants are thrown out from the election; otherwise, the losers are replaced by `copies' of the winner. This is then repeated until there is a single participant remaining, who is declared the winner.

We show that for single peaked preferences, Triadic Consensus converges approximately to the Condorcet winner\footnote{The candidate who would beat any other candidate in a pairwise majority election. In single dimensional spaces, this happens to be the median participant.} with high probability, while only requiring an average of $\sim\log^2 n$ (conjectured to be $\sim\log n$) comparisons per individual. As an illustration, recall our motivating scenario of a city government crowdsourcing a question to its constituencies. Suppose that the city finds ten thousand participants and that the participant preferences are single-peaked. Clearly, looking through all 10000 proposals to perform the simple task of voting for a favorite is impossible. With Triadic Consensus, each participant would make an average of 177 (conjectured to be 13.3) pairwise comparisons for the algorithm to produce a winner. This winner would be between proposals 4950 and 5050 with 95 percent probability and between proposals 4900 and 5100 with 99.99 percent probability. In other words, the winner would be in the top 1 percent of submitted proposals with 95 percent probability and in the top 2 percent of submitted proposals with 99.99 percent probability.

In addition, we show that Triadic Consensus has nice properties for protecting against manipulation. Suppose that the rankings of candidates are induced from an underlying distance metric and suppose that each candidate has a concave utility in that distance. Then Triadic Consensus has a {\it quasi-truthful} Nash equilibrium. Specifically, (see Sect. \ref{sec:quasitruthful}) a Nash equilibrium exists which may not be truthful, but still chooses a winner {\it with the same probability distribution} as if every participant voted truthfully. Surprisingly, we achieve this result by counterintuitively allowing voters to express cyclical preferences (e.g. $a > b$, $b > c$, and $c > a$). 

\subsection{Related Work and Our Contributions}

Given the long history of work on voting theory, it is not surprising that the problems we tackle have been, for the most part, thought about before. Here, we give a brief overview of related work, followed by a summary of our contributions. For in-depth reading, we refer the reader to Brandt et al. \cite{Brandt2012}.

\subsubsection{Voting rule criteria}

One of the earliest criteria introduced for evaluating voting rules is known as the Condorcet criteria, introduced by Marquis de Condorcet\footnote{See Young \cite{Young1988} for a fascinating historical description of the early work of Condorcet.}. It states that if a candidate exists who would win against every other candidate in a majority election, then this candidate should be elected. Unfortunately, such a candidate does not always exist. Since then, many other criteria have been introduced as ways to evaluate voting rules. However, in the surprising result known as Arrow's Impossibility Theorem, Arrow \cite{Arrow1950} proved that there were three desirable criterion that no deterministic voting rule could satisfy. This was expanded by Pattanaik and Peleg \cite{Pattanaik1986} to show that a similar result holds for probabilistic voting rules. 

\subsubsection{Strategic manipulation} 

This sparked a wave of impossibility results, including the classical Gibbard-Satterthwaite Impossibility Theorem. Define a voting rule to be {\it strategy-proof} if it is always in a voter's interest to submit his true preference, regardless of the other voter rankings. Gibbard \cite{Gibbard1973} and Satterthwaite \cite{Satterthwaite1975} independently showed that all deterministic, strategy-proof voting rules must either be dictatorships or never allow certain candidates to win. This was extended to show that only very simple probabilistic voting rules were strategy-proof\cite{Gibbard1977}.

Numerous attempts at circumventing these impossibility result have been made. Bartholdi et al. \cite{Bartholdi1989} first proposed using computational hardness as a barrier against manipulation in elections. However, despite many NP-hardness results on manipulation of voting rules\cite{Faliszewski2010}, it was shown that there do not exist any voting rules that are {\it usually} hard to manipulate\cite{Conitzer2006}. 

Procaccia \cite{Procaccia2010} used the simple probabilistic voting rules of Gibbard \cite{Gibbard1977} to approximate common voting rules in a strategy-proof way, but the approximations are weak and they show that, for many of these voting rules, no strategy-proof approximations can be much stronger. Birrell and Pass \cite{Birrell2011} extended this idea to approximately strategy-proof voting, proving that there exist tight approximations of any voting rule that are close to strategy-proof. Recently, Alon et al. \cite{Alon2011} studied the special case of approval voting when voters and candidates coincide. They show that even though no deterministic strategy-proof mechanism has a finite approximation ratio, a randomized strategy-proof mechanism exists which has a good approximation ratio.

\subsubsection{Communication complexity}

When the number of candidates is large, it is important to study voting rules from the perspective of the burden on voters. Conitzer and Sandholm \cite{Conitzer2005} studied the worst case number of bits that voters need to communicate (e.g. pairwise comparisons) in order to determine the ranking or winner of common voting rules; for many of these voting rules, it was shown that the number of bits required is essentially the same as what is required for reporting the entire ranking. In addition, they showed \cite{Conitzer2002} that for many common voting rules, determining how to elicit preferences efficiently is NP-complete, even when perfect knowledge about voter preferences is assumed. Lu and Boutilier \cite{Lu2011} proposed the idea of reducing communication complexity under approximate winner determination. Though they do not present theoretical guarantees, they propose a regret minimizing algorithm and show significant reductions in communication when run on experimental data sets.

\subsubsection{Single-peaked preferences}

One special case that avoids the many discouraging results above is that of single-peaked preferences\cite{Black1948} (or other domain restrictions). Single-peaked preferences are those for which candidates can be described as lying on a line. Every voter's utility function is peaked at one candidate and drops off on either side. For such preferences, a Condorcet winner always exists and is the candidate who is the median of all voter peaks. This winner can be found by the classical median voting rule, which has each voter state their peak and returns the median of these peaks. It turns out that the median voting rule is both strategy-proof\cite{Moulin1980} and has a low communication complexity of $O(n\log m)$\cite{Escoffier2008}, where $n$ is the number of voters and $m$ is the number of candidates. Conitzer \cite{Conitzer2009} also studies the problem of eliciting voter preferences or the aggregate ranking using comparison queries.

The median voting rule has one weakness: it requires knowledge of an axis, which can make it impractical in practice. First, the {\it algorithm} requires knowledge of the axis in order to pick the median of peaks. When an axis isn't known, Escoffier et al. \cite{Escoffier2008} provides an $O(mn)$ algorithm for finding such an axis with additional queries, but with no strategic guarantees. Second, the {\it voter} also requires knowledge of the axis. In situations where proposals have multiple criterion, but are still single peaked (for example, in a linear combination of the criterion), it may not be obvious to the voter where the axis is. Third, and more subtle, even if an axis is known, it may not be practical to {\it express} a voter's position on this axis. Take, for example, the canonical liberal-conservative axis used to support the single-peaked setting. It is obvious that one extreme of the axis is an absolute liberal and that the other is an absolute conservative. But how would a voter express any position in between? It would not make sense for a voter to express his or her peak as ``seventy percent liberal''.\footnote{Note that he cannot just state his favorite candidate as his peak because this would require looking through all $n$ candidates.}

\subsubsection{Our contributions}\label{sec:contrib}

Triadic Consensus solves the previous problems by eliminating the need for an axis. The only task voters are required to perform is a series of comparisons between two candidates. Likewise, the central algorithm does not require any knowledge about proposal positions. With these properties, we prove the following guarantees (as made precise in Sects. \ref{sec:approx_comm} and \ref{sec:nash}):
\begin{enumerate}
    \item For single-peaked preferences, Triadic Consensus finds a $(1-\epsilon/\sqrt{n})$ approximation of the Condorcet winner with high probability with a communication complexity of $O(\frac{n}{\epsilon^2}\log^2 \frac{n}{\epsilon^2})$, i.e. $\sim n\log^2 n$ (conjectured to be $\sim n\log n$) for a $1 - \frac{1}{\sqrt{n}}$ approximation and $\sim \frac{1}{\epsilon^2}\log^2\frac{1}{\epsilon^2}$ for a $1-\epsilon$ approximation.
    \item For a single-dimensional setting, Triadic Consensus has a quasi-truthful Nash equilibrium when participants have concave utility functions.
\end{enumerate}
These results are especially interesting given that they are coupled with the following novel concepts:
\begin{enumerate}
    \item {\it A localized consensus mechanism for large groups}. We propose Triadic Consensus as an approach for large groups to make decisions using small decentralized decisions among groups of three.
    \item {\it Quasi-truthful voting rules and cyclical preferences}. When each participant is a voter and a candidate, we demonstrate that allowing participants to express cyclical preferences ($a > b$, $b > c$, and $c > a$) can introduce strategies that detect and protect against strategic manipulation. 
\end{enumerate}

\subsubsection{Outline of the paper} Before continuing, we describe the structure of the remaining sections. In Sect. \ref{sec:TriadicConsensus}, we detail Triadic Consensus and introduce the notion of quasi-truthfulness. This is followed by Sect. \ref{sec:approx_comm}, which presents the approximation and communication complexity results, and Sect. \ref{sec:nash}, which describes the quasi-truthfulness results. Finally, Sect. \ref{sec:future} concludes with future directions.

\section{Triadic Consensus and Quasi-truthfulness}\label{sec:TriadicConsensus}

\subsection{Triadic Consensus}\label{subsec:TriadicConsensus}
\begin{algorithm}[H]
    \SetAlgoVlined
    \KwIn{An urn with $k$ labeled balls for each participant $1, 2, \ldots, n$}
    \KwOut{A winning candidate $i$.}
    \While{there is more than one label}{
        Sample three balls (with labels $x, y, z$) uniformly at random with replacement\;
        $w = \text{TriadicVote}(x, y, z)$\;
        \eIf{$w \neq \emptyset$}{
            Relabel all the sampled balls with the winning label $w$\;
        }{
            \tcc{For example, remove the three sampled balls from the urn}
            $\text{TriadicMechanism}(x, y, z)$\; 
        }
    }
    \eIf{at least one ball remains}{
        \Return the id of any remaining ball\;
    }{
        \Return the id of a random ball from the last removed\;
    }
    \caption{Triadic Consensus}
    \label{alg:TriadicA}
\end{algorithm}

\begin{algorithm}[H]
    \SetAlgoVlined
    \KwIn{Candidates $x, y, z$}
    \KwOut{One of $\{x, y, z\}$ if there is a winner, $\emptyset$ otherwise}
    \If{two of more of $x, y, z$ have the same id}{
        \Return the majority candidate\;
    }
    $x$ votes between $y$ and $z$; $y$ votes between $x$ and $z$; $z$ votes between $x$ and $y$\;
    \eIf{each received exactly one vote}{
        \Return $\emptyset$\;
    }{
        \Return the candidate with two votes\;
    }
    
    \caption{TriadicVote}
    \label{alg:Triad}
\end{algorithm}
\vspace{1mm}

Triadic Consensus applies to scenarios where the set of candidates and voters coincide. We use $x$ to refer to both the participant $x$ and the candidate solution that he or she proposes. For $x, y, z \in \{1, 2, \ldots, n\}$, we use $\succ_x$ to denote the ranking of participant $x$ and $y \succ_x z$ to denote that $x$ prefers $y$ over $z$.\footnote{We assume a strict ordering, but it is not hard to generalize the algorithm to ties.}

The best way to understand Triadic Consensus (Alg. \ref{alg:TriadicA}) is to imagine an urn with balls, each of which is labeled by a participant id. The urn starts with $k$ balls for each of the $n$ participants.\footnote{The intuition for $k$ is that it is a tradeoff between approximation and time. Increasing $k$ makes the approximation tighter, but requires more comparisons to converge.} At each step, the algorithm samples three balls uniformly at random (with replacement) and performs a TriadicVote (Alg. \ref{alg:Triad}) on the three corresponding participants.

If the three participants $x$, $y$, and $z$ are unique, the TriadicVote subroutine consists of a single comparison for each of the selected participants: $x$ votes between $y$ and $z$, $y$ between $x$ and $z$, and $z$ between $x$ and $y$. These votes can be distributed in some permutation of $2, 1, 0$ or split $1, 1, 1$. In the first case, the participant who received two votes is returned as the winner. In the second case, a tie (represented as $\emptyset$) is returned. If two or more of the selected ids are the same, i.e. are the same person, then he is automatically returned as the winner.

If a winner was returned from the TriadicVote, then the three balls are relabeled with the winning id and placed back into the urn; otherwise, one of several mechanisms can be applied to resolve the tie. This process is repeated until there is only one participant id remaining, which is declared the winner. A helpful intuition is that the mechanism implemented in the case of a tie can be thought of as a deterrent for manipulation since a three-way tie can never occur if participants are voting truthfully (see Sect. \ref{sec:quasitruthful}).

In our paper, we propose two possible mechanisms, each of which has a quasi-truthful Nash equilibrium. The simplest is Remove (Alg. \ref{alg:mechanismA}), in which the three balls are simply removed. In RepeatThenRemove (Alg. \ref{alg:mechanismB}), the three balls are made to vote again; if there is another three way tie, then they are removed. Surprisingly, repeating the TriadicVote before elimination results in a simpler (and more practical) strategy that is a quasi-truthful Nash equilibrium.

\begin{algorithm}[t]
    \SetAlgoVlined
    \KwIn{Balls $x$, $y$, and $z$}
    Remove the three sampled balls from the urn\;
    \caption{The Remove mechanism}
    \label{alg:mechanismA}
\end{algorithm}
\begin{algorithm}[t]
    \SetAlgoVlined
    \KwIn{Balls $x$, $y$, and $z$}
    $w = \text{TriadicVote}(x, y, z)$\;
    \eIf{$w \neq \emptyset$}{
        Relabel all the sampled balls with the winning label $w$\;
    }{
        $\text{Remove}(x, y, z)$\; 
    }
    \caption{The RepeatThenRemove mechanism}
    \label{alg:mechanismB}
\end{algorithm}





\subsection{Truthfulness and Quasi-truthfulness}\label{sec:quasitruthful}

For our analysis of strategic behavior, we will assume that each individual is represented as a point $x$ in some space $X$ and that his or her preference ranking is induced by a distance metric $d(x, \cdot)$ on $X$. If $d(x, y) \leq d(x, z)$, then $y \succ_x z$; that is, $x$ prefers proposals that are closer to him. Since the individuals voting in a TriadicVote are also the candidates being voted for, there can never be a three-way tie in a truthful vote. Otherwise, all three of $d(x, y) < d(x, z)$, $d(y, z) < d(y, x)$, and $d(z, x) < d(z, y)$ must be simultaneously true, which is impossible so long as $d(\cdot, \cdot)$ satisfies the natural property that $d(x, y) = d(y, x)$.\footnote{If voter rankings allow ties, then a three-way tie could be truthful. In this case, a simple generalization of the TriadicVote that aligns with the intuition of punishing manipulation could be used.}

Consider a TriadicVote between participants $x$, $y$, and $z$. If they vote truthfully, then there are situations when players may be incentivized to deviate as in Ex. \ref{ex:truthful}.

%

\begin{example}\label{ex:truthful}
    Four participants lie in space $X = \mathbb{R}$ at positions $0$, $5$, $6$, and $7$. Suppose participants $0$, $5$, and $7$ are selected for a TriadicVote. Since they are voting truthfully, $0$ votes for $5$, $5$ votes for $7$, and $7$ votes for $5$. As a result, $5$ wins and the resulting urn consists of three balls for $5$ and one for $6$.

Now suppose participant $7$ were to vote strategically for participant $0$. This would result in a tie and all the selected balls would be eliminated, leaving only participant $6$. Clearly, participant $7$ would prefer this second situation.
\end{example}

At this point, we might note that the truthful winner's vote (e.g. $5$) does not change the result and that he can use his vote to disincentivize others from manipulating the TriadicVote. We define any such behavior to be quasi-truthful when it results in the same outcome as that of truthful voting. 
\begin{example}\label{ex:quasi}
    Suppose that the participants of Ex. \ref{ex:truthful} are trying to minimize the expected distance of the winning proposal to their position. Then a quasi-truthful strategy would be for $0$ to vote for $5$, $5$ to vote for $0$ and $7$ to vote for $5$. As in Ex. \ref{ex:truthful}, $5$ wins and the resulting urn consists of three balls for $5$ and one for $6$.

Now suppose participant $7$ deviates from this strategy and votes for $0$. Then participant $0$ gets two votes and he wins. The resulting urn consists of three balls for $0$ and one for $6$, which is clearly worse for participant $7$. Likewise, suppose participant $0$ deviates from this strategy and votes for $7$. Then there is a three-way tie and all selected balls get eliminated. The resulting urn consists of a single ball for $6$, which is clearly worse for participant $0$.
\end{example}
\begin{algorithm}[t]
    \SetAlgoVlined
    \KwIn{Voter $x$, candidates $y, z$}
    \KwOut{One of $\{y, z\}$}
    \eIf{$x$ thinks he should win}{
        \uIf{$y$ would prefer a win for $x$ rather than a three-way tie in a truthful world}{
            \Return $y$\;
        }
        \Else{
            \Return $z$\;
        }
    }{
        \Return a truthful comparison between $y$ and $z$\;
    }
    \caption{Quasi-truthful Nash for the Remove mechanism}
    \label{alg:quasitruthfulnash}
\end{algorithm}

From this example, we get the intuition for Alg. \ref{alg:quasitruthfulnash}, a quasi-truthful Nash for the Remove mechanism. If a participant ($y$ WLOG) is the truthful winner, then he should look for the participant who would prefer a win for him over a removal of all three balls ($x$ WLOG). Then if $y$ votes for $x$, this will disincentivize $x$ from deviating. Since $z$ can only cause $x$ to win by deviating, it seems intuitively bad for him to deviate since the winner will be strictly farther. If all players have concave utilities, it turns out that there always must be a participant that prefers $y$ to win over a removal of all three balls. This intuition is translated into a rigorous proof in Sect. \ref{sec:nash}.

The same idea gives us Alg. \ref{alg:simpquasitruthful}, a quasi-truthful Nash for the RepeatThenRemove mechanism and a more practical strategy to implement. In this strategy, $x$ simply chooses an arbitrary participant to vote for first; if there is a tie, he should then switch his vote. It turns out that the proof structure for this is almost identical to that of Alg. \ref{alg:quasitruthfulnash}. If $y$ happens to vote for $x$ in the first round, then he is safe according to the reasoning of Alg. \ref{alg:quasitruthfulnash}. However, if $y$ happens to vote for $z$, the only thing $z$ can do is to cause a repeat vote, during which $y$ will vote for $x$, bringing us back to the realm of Alg. \ref{alg:quasitruthfulnash}'s strategy.

\begin{algorithm}[t]
    \SetAlgoVlined
    \KwIn{Voter $x$, candidates $y, z$}
    \KwOut{One of $\{y, z\}$}
    \eIf{$x$ thinks he should win}{
        \uIf{it is the first TriadicVote}{
            \tcc{For example, a truthful comparison}
            \Return either of $y$ or $z$\;
        }
        \Else{
            \Return the candidate that he didn't vote for in the first round\;
        }
    }{
        \Return a truthful comparison between $y$ and $z$\;
    }
    \caption{Quasi-truthful Nash for the RepeatThenRemove mechanism}
    \label{alg:simpquasitruthful}
\end{algorithm}

\section{Triadic Consensus approximates the Condorcet winner with low communication complexity}\label{sec:approx_comm}

\subsection{Background: Fixed size urns and urn functions}

The primary idea in proving the results in this section is to reduce the Triadic Consensus urn to previously known results for fixed size urns with urn functions. A fixed size urn contains some number of balls, which are each colored either red or blue. Let $R_t$ and $B_t$ be the number of red and blue balls respectively at time $t$, where $R_t + B_t = n$. Also, let $p_t = \frac{R_t}{n}$ denote the fraction of red balls. At every discrete time $t$, either a red ball is sampled with probability $f(p_t)$, a blue ball is sampled with probability $f(1-p_t)$, or nothing happens with the remaining probability. The function $f : [0,1]\to [0,1]$ is called an urn function and satisfies $0 \leq f(x) + f(1-x) \leq 1$ for $0 \leq x \leq 1$. If a ball was sampled, it is then recolored to the opposite color and placed back into the urn. This process repeats until some time $T$ when all the balls are the same color, i.e. $R_T = n$ or $R_T = 0$. 

We will show in the following section that Triadic Consensus is closely related to fixed size urns with urn function $f(p) = 3p(1-p)^2$. We will then use the following theorems derived from those in Lee and Bruck \cite{Lee2012}\footnote{Theorem \ref{thm:urn_time} requires some algebra that may not be immediately clear from the general theorem stated in the reference. For the convenience of the reader, we include these calculations in Appendix \ref{appsec:supporting}.}.

\begin{theorem}\label{thm:urn}
    Let a fixed size urn start with $R_0$ red balls out of $n$ total balls and have an urn function $f(p) = 3p(1-p)^2$. Let $T$ denote the first time when either $R_T = n$ or $R_T = 0$. Then,
    \[\prob[R_T = n] = \left(\frac{1}{2}\right)^{n-1}\sum\limits_{j = 1}^{R_0}\binom{n-1}{j-1}\]
\end{theorem}
\begin{theorem}\label{thm:urn_time}
    Let a fixed size urn start with $R_0$ red balls out of $n$ total balls and have an urn function $f(p) = 3p(1-p)^2$. Let $T$ denote the first time when either $R_T = n$ or $R_T = 0$. Then,
    \[\mathbb{E}[T] \leq n\ln n + O(n)\]
\end{theorem}

\subsection{Reduction from Triadic Consensus to fixed size urns}\label{sec:approx_proof}

Recall that our results are for the case of single-peaked preferences, for which the candidates can be said to lie on some axis. Every voter's utility is described by a peak on that axis which falls off on either direction. Without loss of generality, we let the participant ids be labeled from one end of the axis to the other, i.e. $1 < 2 < \ldots < n$.
\begin{lemma}\label{lem:median_winner}
    Let $x$, $y$, and $z$ be three unique participants whose peaks lie on an axis such that $x < y < z$. Then the winner of a quasi-truthful TriadicVote($x$, $y$, $z$) must be the median participant $y$.
\end{lemma}
\begin{proof}
    Since $y$ would win in a truthful vote, this follows from the definition of quasi-truthfulness.\qed
\end{proof}

\begin{lemma}\label{lem:prob_winner}
    For single-peaked preferences, let the participant ids be labeled from one end of the axis to the other, i.e. $1 < 2 < \ldots < n$. Color balls with ids $1, 2, \ldots, i$ red and balls with ids $i+1, i+2, \ldots, n$ blue. Then if participants vote quasi-truthfully, Triadic Consensus (for $k = 1$) will produce a red winner with the same probability as that of a fixed size urn with urn function $f(p) = 3p(1-p)^2$.
\end{lemma}
\begin{proof}
    Let $p_r$ and $p_b$ denote the fraction of red and blue balls respectively. Each time three balls are sampled, the median ball must win by Lemma \ref{lem:median_winner}, which implies that the majority color must win. Then we have the following four cases:
    \begin{description}
    \item[{\em Three red}] With probability $p_r^3$, there is no change in colors.
    \item[{\em Two red, one blue}] With probability $3p_r^2 p_b$, one blue ball is recolored red.
    \item[{\em One red, two blue}] With probability $3p_r p_b^2$, one red ball is recolored blue.
    \item[{\em Three blue}] With probability $p_b^3$, there is no change in colors.
    \end{description}
    These are the transition probabilities for a fixed size urn with $i$ red balls, $n-i$ blue balls, and urn function $f(p) = 3p(1-p)^2$. Since every transition probability is identical, the final probability of a red winner must be identical.\qed
\end{proof}

\subsection{Main results}

\begin{theorem}\label{thm:approx_final}
    For single-peaked preferences, let the participant ids be labeled from one end of the axis to the other, i.e. $1 < 2 < \ldots < n$. Then if participants vote quasi-truthfully, Triadic Consensus (for $k=1$) will produce a winner $w$ with probability
    \[\prob[w = i] = \left(\frac{1}{2}\right)^{n-1}\binom{n-1}{i-1}\]
\end{theorem}
\begin{proof}
If balls $1, 2, \ldots, i$ are colored red, then $w \leq i$ iff the winning ball is red. Then applying Theorem \ref{thm:urn} and Lemma \ref{lem:prob_winner}, we get $\prob[w \leq i]$. By subtracting $\prob[w \leq i-1]$ from $\prob[w \leq i]$, we get our final expression.\qed
\end{proof}
    
A similar argument extends the above theorem for general $k$. Using standard probabilistic arguments\cite{Motwani1995}, we get the following corollary.
\begin{corollary}
    Let there be $n$ single-peaked participants and let $w$ denote the winning id after running Triadic Consensus with $k = O(\frac{1}{\epsilon^2}\log \frac{1}{\delta})$. Then assuming that participants vote quasi-truthfully, $w$ will be a $(1-\epsilon/\sqrt{n})$ approximation of the Condorcet winner with probability at least $1 - \delta$.
\end{corollary}

\begin{theorem}
    For single-peaked preferences and quasi-truthful voting, Triadic Consensus has a total communication complexity of $O(kn \log^2 (kn))$.
\end{theorem}
\begin{proof}
    Theorem \ref{thm:urn_time} is an upper bound on the expected time it takes to halve the number of remaining participants (since we can color half the participants red and half blue). For $kn$ balls, this gives us less than or equal to $kn\ln(kn) + O(kn)$ expected time to halve the participants. Each time the urn converges to a single color, we can recolor half the remaining participants and repeat. After $\log n$ rounds, we will be done.\qed
\end{proof}

The above theorem is an upper bound on the communication complexity. In reality, at each recoloring, the balls will not be split evenly between the two colors. Based on this intuition and simulations, we conjecture that the communication complexity is only $O(kn\log kn)$.
\begin{conjecture}
    For single-peaked preferences and quasi-truthful voting, Triadic Consensus has a total communication complexity of $O(kn \log kn)$.
\end{conjecture}

\subsection{How bad is $O(\sqrt{n})$ error?}\label{sec:how_bad}

We would like to point out that for crowdsourcing applications, participants can often be viewed as noisy samples from some underlying distribution. Because of this, even if we were to find the exact Condorcet winner of the noisy sample, this may not significantly improve the variance.

For example, suppose that the participants are drawn independently and uniformly from $[0, 1]$ so that the true Condorcet candidate would be one in position $\frac{1}{2}$. Now suppose that $\frac{1}{2}$ lies between the $k$-th and $k+1$-th sampled participant. Then $k$ is clearly binomially distributed, which is the same distribution as the case of Triadic Consensus. In other words, the exact Condorcet winner of the sampled participants has a standard deviation of $\frac{1}{2}\sqrt{n}$ participants between him and the true Condorcet winner of the underlying distribution. Since the approximate Condorcet winner produced by Triadic Consensus also $O(\sqrt{n})$ standard deviation, they have the same order of error.

\subsection{Triadic Consensus eliminates outliers quickly}\label{sec:outliers}

Triadic Consensus has the intuition of quickly eliminating outliers since each participant needs to convince two other participants to vote for him in order to win. We will give two observations supporting this intuition. First, we consider two seemingly powerful algorithms that have access to side information that allows them to directly eliminate outliers in various ways. Despite their use of this knowledge, they choose a winner with the exact same probability distribution as that of Triadic Consensus. Second, we compare Triadic Consensus to a non-Triadic random sampling algorithm and show that Triadic Consensus does much better at picking a central candidate for each round.

\subsubsection{Two seemingly powerful algorithms} Consider algorithms RemoveRandomExtreme (Alg. \ref{alg:removerandom}) and ContractExtremes (Alg. \ref{alg:contract}). Each of these algorithms assume a set of participants which have single-peaked preferences and have some limited access to the two extremes (the participants with the leftmost and rightmost peaks). In RemoveRandomExtreme, one of the two extreme balls is randomly chosen and thrown out. In ContractExtremes, a random ball is chosen, and both of the extreme balls are moved to this chosen ball. We show that they choose a winner with the same probability distribution as that of Triadic Consensus.

\begin{theorem}
    For single-peaked preferences, let the participant ids be labeled from one end of the axis to the other, i.e. $1 < 2 < \ldots < n$. Then RemoveRandomExtreme and ContractExtremes will produce a winner $w$ with probability
    \[\prob[w = i] = \left(\frac{1}{2}\right)^{n-1}\binom{n-1}{i-1}\]
\end{theorem}
\begin{proof}
We use the same coloring technique as that of the proof for Lemma \ref{lem:prob_winner}
\end{proof}
\begin{algorithm}[t]
\SetAlgoVlined
\KwIn{An urn with a ball for each participant $1, 2, \ldots, n$, each of whom lie on an axis}
\KwOut{A winning candidate $i$.}
\While{there is more than one label}{
    Randomly sample one of the two (left or right) extreme balls\;
    Toss out the sampled ball\;
}
\Return remaining label\;
\caption{RemoveRandomExtreme}
\label{alg:removerandom}
\end{algorithm}
\begin{algorithm}[t]
\SetAlgoVlined
\KwIn{An urn with a ball for each participant $1, 2, \ldots, n$, each of whom lie on an axis}
\KwOut{A winning candidate $i$.}
\While{there is more than one label}{
    Randomly sample one candidate\;
    Move both of the two (left and right) extreme balls to the sampled ball\;
}
\Return remaining label\;
\caption{ContractExtremes}
\label{alg:contract}
\end{algorithm}

\subsubsection{Hot-or-Not Consensus}

Consider Hot-or-Not Consensus, in which two balls are randomly chosen as candidates and one single ball is randomly chosen as the voter. The voter then votes between the two chosen balls and the two candidate balls are replaced with the winning candidate. We do a one-step comparison of these two algorithms given a continuous distribution of participants. It turns out that the naive Hot-or-Not Consensus can be thought of as a mix between Triadic Consensus and the (really bad) method of randomly picking a candidate.\footnote{This is only an intuition based on a one-step comparison and should not be interpreted as a comparison of their final approximation values.}
\begin{theorem}
Let a continuous distribution of voters be uniformly distributed between zero and one. Let $g_{\text{Hot-or-Not}}(x)$ and $g_{\text{Triadic}}(x)$ be the probability density of $x$ being the next winning candidate in Hot-or-Not and Triadic Consensus respectively. Then, $g_{\text{Triadic}} = 6x(1-x)$ and $g_{\text{Hot-or-Not}} = 3x(1-x) + \frac{1}{2}$. In particular,
\[g_{\text{Hot-or-Not}} = \frac{1}{2}g_{\text{Triadic}}+\frac{1}{2}g_{\text{Unif}}\]
where $g_{\text{Unif}}$ is the uniform distribution over the interval $[0,1]$.
\end{theorem}
\begin{proof}
    For uniformly distributed participants, we have the density function $f(x) = 1$ and cumulative density function $F(x) = x$. In Triadic Consensus, $x$ wins if he is selected along with a candidate to the left and right of him. Then, we have 
\[g_{\text{Triadic}}(x) = 3!f(x)F(x)(1-F(x)) = 6x(1-x)\]
In Hot-or-Not Consensus, $x$ wins against $y$ only if the voter $z$ is closer to $x$ than $y$. Then,
\begin{align*}
    g_{\text{Hot-or-Not}}(x) &= 2f(x)\int_0^x f(y)\left(1-F\left(\frac{x+y}{2}\right)\right)\mathrm{d}y + 2f(x) \int_x^1 f(y)F\left(\frac{x+y}{2}\right)\mathrm{d}y\\
                          &= 2\left[\int_0^x \left(1-\frac{x+y}{2}\right)\mathrm{d}x + \int_x^1 \left(\frac{x+y}{2}\right)\mathrm{d}x\right] = 3x(1-x) + \frac{1}{2}
\end{align*}
\end{proof}

\subsection{Simulations in general spaces show promise}\label{sec:simulations}

Even though our approximation and communication complexity results only hold for single-peaked preferences, we believe that Triadic Consensus has strong properties for more complex spaces. To demonstrate this, we consider two classes of preferences that are induced by points in a two dimensional Euclidean space. In each of these cases, the rankings are clearly not single-peaked, but simulations show strong results for both the approximation and the communication complexity. The first example is a straightforward generalization to points that are laid out in a grid.

\begin{example}
$n = m^2$ voters are placed on the points $(0,0), (0,1), \ldots, (m-1, m-1)$ to form a $m\times m$ grid. The Condorcet winner in this scenario is at the median point $(\frac{m}{2}, \frac{m}{2})$. From the simulation results below, we can see that Triadic Consensus picks winners that are closely distributed around the winner. The average number of votes each voter casts is $\sim O(\log n)$.

\begin{table}\label{tab:grid}
\begin{center}
  \begin{tabular}{l|cccc}
& mean winner & $\sigma$ of winner & mean votes/voter & $\sigma$ of votes/voter\\
\hline
5x5 & (1.96, 1.97) & .953 & 1.966 & 0.425 \\
10x10 & (4.45, 4.5) & 1.157 & 3.405 & 0.517 \\
20x20 & (9.61, 9.47) & 1.236 & 4.618 & 0.417 \\
40x40 & (19.64, 19.53) & 1.594 & 6.056 & 0.368 \\
80x80 & (39.57, 39.69) & 1.555 & 7.293 & 0.324
\end{tabular}
\end{center}
\caption{Simulation results for 100 iterations of Triadic Consensus on a grid}
\end{table}
\end{example}
In the second case, we try to design a difficult scenario by densely populating the perimeter of a circle and adding a single point at its center, who is the Condorcet winner.
\begin{example}
$n$ voters are placed uniformly around a circle (in the plane) with radius $1$ and centered at $(0,0)$. A single voter is placed at the point $(0,0)$. The Condorcet winner in this scenario is the point $(0,0)$. Surprisingly, we find that even as the number of points increases on the perimeter, the probability of the randomized algorithm selecting the single Condorcet winner still remains non-trivial. If this probability does remain above some constant, then we can use standard probabilistic methods to show that repeating Triadic Consensus a small number of times and picking the majority winner will result in $(0,0)$ with arbitrarily high accuracy. The average number of votes each voter casts is $\sim O(\log n)$.
\begin{table}
\begin{center}
\begin{tabular}{l|cccc}
& \% times (0,0) wins & mean votes/voter & $\sigma$ of votes/voter\\
\hline
25 & 0.368 & 2.225 & 0.513 \\
100 & 0.338 & 4.189 & 0.775 \\
400 & 0.305 & 6.628 & 1.185 \\
1600 & 0.295 & 9.224 & 1.696 \\
6400 & 0.302 & 11.764 & 2.321
\end{tabular}
\end{center}
\caption{Simulation results for 1000 iterations of Triadic Consensus on a circle with a single point in the center}
\end{table}
\end{example}
\section{Triadic Consensus has a quasi-truthful Nash equilibrium for concave utilities}\label{sec:nash}

To discuss strategic behavior, we need to define the utilities for each participant. Let $U_x(y)$ denote the utility that $x$ gets from a proposal $y$. The utility $x$ derives from $y$ depends on the distance from $x$ to $y$, i.e. $U_x(y) = f_x(d(x, y))$, where $f(\cdot)$ must be decreasing in distance so that $U_x(y) > U_x(z)$ whenever $y \succ_x z$. We say that a participant $x$ has a concave utility function if $f_x(\cdot)$ is a concave function.

\begin{theorem}
    If all participants have concave utility functions, then Algs. \ref{alg:quasitruthfulnash} and \ref{alg:simpquasitruthful} are quasi-truthful Nash equilibria for Triadic Consensus when using the Remove and RepeatThenRemove mechanisms respectively.
\end{theorem}
\begin{proof}
  We prove our main result with the following proof by induction. Since the proofs for the Remove mechanism and the RepeatThenRemove mechanism are almost identical (see \ref{sec:quasitruthful}), we will refer solely to the Remove mechanism for simplicity.

\paragraph{Base Case:} Alg. \ref{alg:quasitruthfulnash} is a Nash equilibrium for $n = 1, 2, 3$ balls (Lemma \ref{lem:basecase}).

\paragraph{Inductive Step:} Assume that Alg. \ref{alg:quasitruthfulnash} is a Nash equilibrium for $n-3$ balls. Now consider a participant $x$ who is considering deviating from Alg. \ref{alg:quasitruthfulnash} in an urn with $n$ balls:
\begin{enumerate}
    \item For any TriadicVote with participants $x < y < z$ in an urn with $n$ balls, if $y$ votes for $x$, then by the definition of the strategy and the fact that one of $x$ and $z$ must prefer $y$ to a three-way tie (Lemma \ref{lem:existence}), we know that $x$ must prefer $y$ to win over a three-way tie, which means $x$ should not deviate.
    \item For any TriadicVote with participants $x < y < z$ in an urn with $n$ balls, if $y$ votes for $z$, then given that the previous statement is true, we show that $x$ should prefer a win for $y$ over a win for $z$ (Lemma \ref{lem:coupling}). This is done by defining a comparison relation between urns that formalizes this intuition that participants should prefer closer balls. With this definition, we can define a coupling of two urns: one in which $x$ plays an optimal strategy, and one in which $x$ always plays according to Alg. \ref{alg:quasitruthfulnash}. We show that for every coupled history, the urn from Alg. \ref{alg:quasitruthfulnash} does at least as well as the optimal urn in expected utility. This means that Alg. \ref{alg:quasitruthfulnash} is also an optimal strategy for $x$ in this case.
\end{enumerate}
By carrying out the Inductive Hypothesis, we get our result for all $n$.\qed
\end{proof}

\subsection{Supporting Lemmas}

\begin{lemma}\label{lem:basecase}
    Alg. \ref{alg:quasitruthfulnash} is a Nash equilibrium for Triadic Consensus with the Remove mechanism when $n = 1$, $2$, or $3$ balls.
\end{lemma}
\begin{proof}
    This is trivially true for $n = 1$ and $2$ since no votes take place. For $n=3$, suppose that the three participants are $x < y < z$. In this case, the only situation when participants cast votes is when TriadicVote is performed with all three unique participants. After such a situation occurs, there will either be a winner or all balls will be eliminated and no further votes take place. Therefore, our analysis can be constrained to this single TriadicVote.

    If participants vote according to Alg. \ref{alg:quasitruthfulnash}, we know that $y$ will be the winner since $x$ and $z$ both vote for him. Suppose $y$ votes for $x$ WLOG. Then if $z$ deviates, $x$ will win, which is clearly suboptimal. If $x$ deviates, then there is a three-way tie and all are eliminated, resulting in a uniformly random winner. 

    The difference in utility lost for $x$ by deviating is $\Delta U_x = U_x(y) - \frac{1}{3}(U_x(x) + U_x(y) + U_x(z))$. Letting $d_1$ be the distance between $x$ and $y$ and $d_2$ the distance between $y$ and $z$, we have $\Delta U_x = \frac{1}{3}(f_x(d_1) - f_x(0)) - \frac{1}{3}(f_x(d_1+d_2) - f_x(d_1))$ and 
    \begin{align*}
        \Delta U_x \geq 0 &\iff f_x(d_1) - f_x(0) \geq f_x(d_1+d_2) - f_x(d_1)\\
                          &\iff \frac{f_x(d_1) - f_x(0)}{f_x(d_1+d_2) - f_x(d_1)} \leq 1
    \end{align*}
    Similarly, 
    \[\Delta U_z \geq 0 \iff \frac{f_z(d_2) - f_z(0)}{f_z(d_1+d_2) - f_z(d_2)} \leq 1\]
    For concave, monotonically non-increasing $f_x$ and $f_z$, we know that (detailed in the long version\cite{Goel2012}):
    \[\frac{f_x(d_1) - f_x(0)}{f_x(d_1+d_2) - f_x(d_1)} \leq \frac{d_1}{d_2} \hspace{.2in} \text{ and } \hspace{0.2in} \frac{f_z(d_2) - f_z(0)}{f_z(d_1+d_2) - f_z(d_2)} \leq \frac{d_2}{d_1}\]
    But then, at least one of $\frac{d_1}{d_2}$ or $\frac{d_2}{d_1}$ is less than or equal to $1$, which means that at least one of $\Delta U_x$ and $\Delta U_z$ is greater than or equal to $0$ and prefers a win for $y$ over a three-way tie. By the definition of Alg. \ref{alg:quasitruthfulnash}, $y$ will vote for this person when he exists. Therefore, since $y$ voted for $x$, we know $\Delta U_x \geq 0$, which concludes the proof.\qed
\end{proof}
 
\begin{lemma}\label{lem:existence}
    Assume that Alg. \ref{alg:quasitruthfulnash} is a Nash equilibrium for any configuration of $n-3$ balls. Then for a TriadicVote among participants $x < y < z$ in an urn with $n$ balls, at least one of $x$ or $z$ prefers a win for $y$ over a three-way tie, so long as they both have concave utilities.
\end{lemma}
\begin{proof}
    Because of space constraints, we will only outline the proof here, leaving the notation and algebra for Appendix \ref{appsec:existence}. The proof has two parts:

    Part A. Suppose all balls are positioned somewhere between $x$ and $z$, i.e. in the interval $[x, z]$. Then, if $x$ and $z$ have concave utility functions, at least one of $x$ and $z$ prefers a win for $y$ over a three-way tie. The proof for this statement is similar to the one in Lemma \ref{lem:basecase}, albeit more complex. 

    Part B. For any configuration of $n$ balls, moving any ball at position $x$ leftwards and moving any ball at position $z$ rightwards can only increase both $\Delta U_x$ and $\Delta U_z$. Put another way, given any configuration, we can move all balls left of $x$ to $x$ and all balls right of $z$ to $z$, while only decreasing $\Delta U_x$ and $\Delta U_z$. Once moved in this way, the configuration of balls falls under the jurisdiction of Part 1, which states that at least one of $\Delta U_x$ and $\Delta U_z$ is greater than or equal to $0$. Therefore, the same participant in the original configuration must also have a positive $\Delta U$, which means he prefers a win for $y$ over a three-way tie.\qed
\end{proof}
    
For the final lemma, we require the following definition.
\begin{definition}\label{def:urn_domin}
    Given two urns $R$ and $S$, each with $n$ balls, number the balls in $R$ from left to right as $r_1, r_2, \ldots, r_n$ and number the balls in $S$ from left to right as $l_1, l_2, \ldots, l_n$. Then $R$ $x\text{-dominates}$ $S$ if
\begin{align*}
    s_i \leq r_i &\text{ for } r_i < x\\ 
    s_i = r_i    &\text{ for } r_i = x\\
    s_i \geq r_i &\text{ for } r_i > x
\end{align*}
\end{definition}
\begin{lemma}\label{lem:coupling}
    Assume that Alg. \ref{alg:quasitruthfulnash} is a Nash equilibrium for any configuration of $n-3$ balls. Then for a TriadicVote among participants $x < y < z$ in an urn with $n$ balls, if $y$ votes for $z$ (WLOG), $x$ does not benefit by voting strategically for $z$.
\end{lemma}
\begin{proof}
Our proof strategy will be to use a coupling argument. Let $\text{OPT}$ denote the optimal strategy for $x$. We consider two urns $R$ and $S$. In urn $R$, $x$ plays according to Alg. \ref{alg:quasitruthfulnash}. In urn $S$, $x$ plays according to $\text{OPT}$, the strategy that maximizes his expected utility. We couple the TriadicVote's of these urns in the following way:
\begin{enumerate}
    \item Let $r_1, r_2, \ldots, r_n$ denote the balls in urn $R$ as indexed from leftmost position to rightmost position. Let $s_1, s_2, \ldots, s_n$ denote the balls in urn $S$ as indexed from leftmost position to rightmost position.
    \item Then for every TriadicVote, when balls $r_i, r_j, r_k$ are randomly drawn from urn $R$, balls $s_i, s_j, s_k$ will be drawn from urn $S$.
\end{enumerate}
Suppose $R$ $x\text{-dominates}$ $S$ and then each undergoes a coupled TriadicVote where balls $r_i < r_j < r_k$ are selected from $R$ and $s_i < s_j < s_k$ are selected from $S$. After they vote, we show that the resulting urns $R'$ and $S'$ must still satisfy $R'$ $x\text{-dominates}$ $S'$. By the coupling rule, this is trivially true when 1) $x$ is not selected, 2) $x$ is represented in two or more balls, and 3) $x$ is the middle participant. This is because $x$ either does not vote or cannot affect the result in these cases (remember that all other participants are voting according to Alg. \ref{alg:quasitruthfulnash}). The only remaining case is when $x$ is one of the side participants ($s_i$ WLOG). In this case, $r_j$ wins in urn $R$ since $x$ plays according to Alg. \ref{alg:quasitruthfulnash} in this urn. Suppose $s_j$ voted for $s_k$. Then regardless of who $x$ votes for, one of $s_j$ or $s_k$ must win, both of which will still satisfy $R'$ $x\text{-dominates}$ $S'$. Now suppose $s_j$ voted for $s_i$. Then $x$ could eliminate all three participants by voting for $s_k$. However, by Lemma \ref{lem:existence} and the definition of Alg. \ref{alg:quasitruthfulnash}, this would be suboptimal, which means that $x$ cannot play this strategy in urn $S$. Therefore, $R'$ $x\text{-dominates}$ $S'$.

Finally, we note that before any TriadicVote's take place, $R$ and $S$ are identical, i.e. $R$ $x\text{-dominates}$ $S$. Then, the winner of $R$ must also $x\text{-dominate}$ the winner of $S$, which means that urn $R$ is better for $x$ in every coupled history.\qed
\end{proof}

\section{Future Directions}\label{sec:future}

There are many future directions for this work. 

\paragraph{Triadic Consensus} For the algorithm itself, the primary problem that begs to be worked on is an analysis for higher dimensional or even non-Euclidean spaces. It is an open question whether Triadic Consensus achieves low communication complexity for general preference profiles and whether it has nice properties of convergence. The authors believe that there is something interesting that can be said here. Similarly, it would also be exciting to extend the work on quasi-truthfulness to higher dimensional spaces. The authors do not believe that a naive extension will suffice; however, it seems possible that probabilistic strategies coupled with other punishments for manipulation will be able to achieve this goal.

\paragraph{Truthful voting rules} When participants are voters and candidates, we have indicated that manipulation can often be detected. It would be interesting to use this idea, possibly along with the theme of triads, quasi-truthfulness, and cyclic preferences, to design truthful voting rules. For example, one could imagine the following variant of the Borda count: for each of the $\binom{n}{3}$ triads, add one point to the score of the winner\footnote{The Borda count is equivalent to giving the winner two points, the next highest scoring participant one point, and the loser zero points.}. 

\paragraph{Communication complexity} Another exciting problem is to make new approximate and randomized voting rules that have low communication complexity. In particular, it would be useful to have a voting rule where the maximum number of comparisons per voter is small (say, $O(\log n)$). In Triadic Consensus, only the average number of comparisons is small, which may still prevent it from being widely applicable to large internet crowdsourcing applications. 

\paragraph{Consensus mechanisms} On the direction of group consensus mechanisms, one possible extension of this work is to bring it outside of voting. Namely, rather than having the randomly selected triads {\it vote}, it would be interesting to analyze other sorts of dynamics that are more collaborative or game-theoretic.


\paragraph{Urn voting rules} It would also be interesting to study generalized urn voting rules. This could include different ball replacement schemes or even more elaborate generalizations. For example, balls could be labeled with participant {\it and} proposal ids so that only proposal ids are changed after a TriadicVote. Such urn voting rules are interesting because they can be interpreted as local decisions made by small groups of people and also provide a natural framework for studying (non-trivial) probabilistic voting rules.

\section*{Acknowledgements}
    The authors would like to thank Pranav Dandekar for helpful discussions and Vincent Conitzer for pointing us to several pieces of related literature and posing the idea of lowering communication complexity with approximations. This research was supported in part by NSF grants 0904325 and 0947670. David Lee was also supported in part by an NSF Graduate Research Fellowship.

\bibliographystyle{splncs}
\bibliography{TriadicArXiv}

\appendix

\section{Proof of Lemma \ref{lem:existence}}\label{appsec:existence}
We restate Lemma \ref{lem:existence} here for the convenience of the reader.
\begin{lemma}
    Assume that Alg. \ref{alg:quasitruthfulnash} is a Nash equilibrium for any configuration of $n-3$ balls. Then for a TriadicVote among participants $x < y < z$ in an urn with $n$ balls, at least one of $x$ or $z$ prefers a win for $y$ over a three-way tie, so long as they both have concave utilities.
\end{lemma}
As described in Section \ref{sec:nash}, this proof will be broken up into two parts, which we will prove below. Each of these lemmas assume the inductive hypothesis that Alg. \ref{alg:quasitruthfulnash} is a Nash equilibrium for any configuration of $n-3$ balls.

\subsection{Preliminaries}
Let $P$ denote the urn resulting from a win for $y$ and $Q$ denote the urn resulting from a three-way tie. $P$ has $n$ balls, while $Q$ has $n-3$ balls. 

Index the balls in urn $P$ as $b_1, b_2, \ldots, b_n$ from the leftmost participant position to the rightmost and let the balls in urn $Q$ be indexed identically. Let $b_l, b_{l+1}, b_{l+2}$ denote the three balls in urn $P$ that are not in urn $Q$, so that urn $Q$ has balls $b_1, \ldots, b_{l-1}, b_{l+3}, \ldots, b_n$. Note that these three balls can be indexed consecutively because all of these missing balls are at the same location (participant $y$).

For each ball $b_i$, there is some probability that the ball wins in urn $P$ and some probability that the ball wins in urn $Q$.\footnote{Technically, it is not clear what it means for one ball in position $p$ to win over another ball in position $p$; we let the winning probability of the $i$-th ball of $n$ total balls be $\propto \binom{n-1}{i-1}$ since this is convenient and still produces the correct participant winning probabilities.} Since $b_l, b_{l+1}, b_{l+2}$ don't exist in urn $Q$, their probability of winning there is simply $0$.

Let $\Delta p(b_i) = \prob[b_i \text{ wins in urn $P$}] - \prob[b_i \text{ wins in urn $Q$}]$. 

We will use $U_x(b_i)$ to denote the utility of a $b_i$ win for $x$. As before, $U_x(b_i) = f_x(d(x, b_i))$. If $f_x$ is concave, then $x$ is said to have a concave utility function. Let $U_x^P$ denote the expected utility for $x$ from quasi-truthful voting in urn $P$. Finally, we use $\Delta U_x = U_x^P - U_x^Q$ to denote the difference in expected utility in quasi-truthful voting for urn $P$ and compared to urn $Q$. 

With this notation, our proof essentially boils down to proving that at least one of $\Delta U_x \geq 0$ or $\Delta U_z \geq 0$ holds. We start with a Lemma which we will need for the further arguments. Roughly speaking, it states that there is an interval from a ball on participant $y$ to the median ball for which the probability of winning in $P$ is greater than $Q$. All balls outside this interval have a lower probability of winning in $P$ than in $Q$.

\begin{lemma}\label{lem:delta_add_tokens}
    Suppose Triadic Consensus is run on the urns $P$ and $Q$. Then for quasi-truthful voting,
    \begin{align*}
        \Delta p(b_i) > 0 &\text{ if $\min(l, n/2) \leq i \leq \max(l+2, n/2)$}\\
        \Delta p(b_i) < 0 &\text{ otherwise}
    \end{align*}
\end{lemma}
\begin{proof}
    Recall that $\Delta p(b_i) = \prob[b_i \text{ wins in urn $P$}] - \prob[b_i \text{ wins in urn $Q$}]$. From Theorem \ref{thm:approx_final}, we have that a quasi-truthful strategy in urns $P$ and $Q$ give,
    \begin{align*}
        \prob[b_i \text{ wins in urn $P$}] &= \left(\frac{1}{2}\right)^{n-1}\binom{n-1}{i-1}\\
        \prob[b_i \text{ wins in urn $Q$}] &= \begin{cases}
            \left(\frac{1}{2}\right)^{n-4}\binom{n-4}{i-1} &\text{ if } i \leq l-1\\
            \left(\frac{1}{2}\right)^{n-4}\binom{n-4}{i-4} &\text{ if } i \geq l+3\\
            0                                              &\text{ if } i = l, l+1, l+2
        \end{cases}
    \end{align*}
    Note that for $i < l$,
    \begin{align*}
        \Delta p(b_i) &= \left(\frac{1}{2}\right)^{n-1}\left[\frac{(n-1)!}{(i-1)!(n-i)!} - 8\frac{(n-4)!}{(i-1)!(n-i-3)!}\right]\\
                      &= \left(\frac{1}{2}\right)^{n-1}\frac{(n-4)!}{(i-1)!(n-i)!}\underbrace{\left[(n-1)(n-2)(n-3) - 8(n-i)(n-i-1)(n-i-2)\right]}_{f(i)}
    \end{align*}
    Since $f(i)$ is monotonically increasing in $i$, then by observing that $f\left(\frac{n}{2} - 1\right) < 0$ and $f\left(\frac{n}{2}\right) > 0$, we have
    \begin{align*}
        \Delta p(b_i) \text{ is }\begin{cases}
            < 0 \text{ if $i < \min(l, n/2)$}\\
            > 0 \text{ if $n/2 \leq i < l$}
        \end{cases}
    \end{align*}
    Similarly, for $i > l+2$, we can use an analogous argument (or apply symmetry) to claim that,
    \begin{align*}
        \Delta p(b_i) \text{ is }\begin{cases}
            < 0 \text{ if $i > \max(l+2, n/2)$}\\
            > 0 \text{ if $l+2 < i \leq n/2$}
        \end{cases}
    \end{align*}
    Finally, it is clear that $\Delta p(b_i) > 0$ for $i = l, l+1, l+2$, so we are done.
\end{proof}

\subsection{Part A: If all $b_i$ lie in $[x, z]$}

\begin{lemma}
    Assume that the inductive hypothesis holds for $n-3$. Then for a TriadicVote among participants $x < y < z$ in an urn with $n$ balls, each of which lie in $[x, z]$, at least one of $x$ or $z$ prefers a win for $y$ over a three-way tie, so long as they both have concave utilities.
\end{lemma}
    Recall that our lemma boils down to proving that at least one of $\Delta U_x \geq 0$ and $\Delta U_z \geq 0$ is true. Since the proof is very notation heavy, we first sketch the proof with an example.
\begin{example}
    Let $P$ be an urn with five balls: $b_1$ is located at position $x$; $b_2$, $b_3$, and $b_4$ are located at position $y$; and $b_5$ is located at position $z$. Then urn $Q$ is an urn with the two balls $b_1$ and $b_5$.

    We know (see Theorem \ref{thm:approx_final}) that for balls $b_1$, $b_2$, $b_3$, $b_4$, and $b_5$, $\prob[b_i \text{ wins in urn $P$}]$ is $\frac{1}{16}, \frac{4}{16}, \frac{6}{16}, \frac{4}{16}$, and $\frac{1}{16}$ respectively; $\prob[b_i \text{ wins in urn $Q$}]$ is $\frac{1}{2}, 0, 0, 0$, and $\frac{1}{2}$, respectively; which means that $\Delta p(b_i)$ is $-\frac{7}{16}, \frac{4}{16}, \frac{6}{16}, \frac{4}{16}$, and $-\frac{7}{16}$, respectively. Then,
\begin{align*}
    \Delta U_x &= \Delta p(b_1)U_x(b_1) + \left[\sum_{i=2}^4 \Delta p(b_i)U_x(b_i)\right] + \Delta p(b_5)U_x(b_5)\\
    &= -\frac{7}{16} f_x(d(x,x)) + \frac{14}{16} f_x(d(x, y)) - \frac{7}{16} f_x(d(x, z))\\
    &= -\frac{7}{16} [f_x(d(x,x)) - f_x(d(x,y))] - \frac{7}{16} [f_x(d(x, z)) - f_x(d(x, y))]
\end{align*}
Since $\frac{7}{16}[f_x(d(x,x)) - f_x(d(x, y))] \geq 0$, we have that
\[\Delta U_x \geq 0 \iff \frac{-\frac{7}{16} [f_x(d(x, z)) - f_x(d(x, y))]}{-\frac{7}{16} [f_x(d(x,x)) - f_x(d(x,y))]} \geq 1\]
Using similar arguments and the fact that $\frac{7}{16}[f_z(d(z,x)) - f_z(d(z, y))] \leq 0$, we have that
\[\Delta U_z \geq 0 \iff \frac{-\frac{7}{16} [f_z(d(z, z)) - f_z(d(z, y))]}{-\frac{7}{16} [f_z(d(z,x)) - f_z(d(z,y))]} \leq 1\]
By concavity and Lemma \ref{lem:concave_ineq} in Appendix \ref{appsec:supporting},
\begin{align*}
    \frac{f_x(d(x, z)) - f_x(d(x, y))}{f_x(d(x,x)) - f_x(d(x,y))} &\geq \frac{(z-x) - (y-x)}{(x-x) - (y-x)} = \frac{z-y}{x-y}\\
                                                                &= \frac{(z-z) - (z-y)}{(z-x) - (z-y)} \geq \frac{f_z(d(z, z)) - f_z(d(z, y))}{f_z(d(z,x)) - f_z(d(z,y))}
\end{align*}
This means that if $\frac{z-y}{x-y} \geq 1$, then $\Delta U_x \geq 0$. Otherwise, if $\frac{z-y}{x-y} \leq 1$, then $\Delta U_z \geq 0$. Obviously, one of these must be true, so we are done for this example.
\end{example}

\begin{proof}

We will now prove our lemma for the general case. The argument structure is exactly the same. Let $A = \min(l, n/2)$ and $Z = \max(l+2, n/2)$. Recall that these are the leftmost and rightmost balls for which $\Delta p(b_i) > 0$. All other balls must have $\Delta p(b_i) < 0$. Then we can separate the expression for $\Delta U_x$ into three summations,
\begin{align*}
    \Delta U_x = \sum_{i=1}^{A-1} \Delta p(b_i) U_x(b_i) + \sum_{i=A}^Z \Delta p(b_i) U_x(b_i) + \sum_{i=Z+1}^n \Delta p(b_i) U_x(b_i)
\end{align*}
Note that $\sum_{i=1}^{A-1} \Delta p(b_i) + \sum_{i=Z+1}^n \Delta p(b_i) = -\sum_{i=A}^Z \Delta p(b_i)$. Then we can partition up the mass of $\Delta p(b_j)$ so that we can get one-to-one correspondence of masses corresponding to $j \in [A, Z]$ and $j \not\in [A, Z]$. Let the masses corresponding to $i=1, 2, \ldots, A-1$ be denoted by $p_1, p_2, \ldots, p_u$ and let the masses corresponding to $Z_1, Z_2, \ldots, n$ be denoted by $q_1, q_2, \ldots, q_v$. In other words, $p_{1\ldots u}$, $q_{1\ldots v}$ are chosen so that $\sum_{i=1}^{A-1} \Delta p(b_i) = \sum_{i=1}^u p_i$, $\sum_{i=Z+1}^n \Delta p(b_i) = \sum_{i=1}^v q_i$, and 
\begin{align*}
    \Delta U_x = \sum_{i=1}^{u} p_i [U_x(b_{g_1(i)}) - U_x(b_{h_1(i)})] + \sum_{i=1}^v q_i [U_x(b_{g_2(i)}) - U_x(b_{h_2(i)})]
\end{align*}
where $g_1:[1..u] \to [1..A-1]$, $g_2:[1..v] \to [Z+1..n]$, $h_1:[1..u] \to [A..Z]$, and $h_2:[1..v] \to [A..Z]$.

Note that since balls are indexed left to right and all lie within $[x, z]$, then balls indexed $i \in [1, A)$ (e.g. $b_{g_1(\cdot)}$) are closer to $x$ than those indexed $i \in [A, Z]$ (e.g. $b_{h_1(\cdot)}$ and $b_{h_2(\cdot)}$), which are closer than those indexed $i \in (R, n]$ (e.g. $b_{g_2(\cdot)}$). Therefore, we have
\begin{align*}
    U_x(b_{g_1(i)}) - U_x(b_{h_1(i)}) \geq 0\text{ and } U_x(b_{g_2(i)}) - U_x(b_{h_2(i)}) \leq 0
\end{align*}
Combining these, we get that 
\begin{align}
    \Delta U_x \geq 0 \iff \frac{\sum_{i=1}^{u}p_i[U_x(b_{g_1(i)}) - U_x(b_{h_1(i)})]}{\sum_{i=1}^{v}q_i[U_x(b_{h_2(i)}) - U_x(b_{g_2(i)})]} \geq 1\label{eqn:utilityx}
\end{align}
since $q_i[U_x(b_{h_2(i)}) - U_x(b_{g_2(i)})] \geq 0$. Similarly, we have
\begin{align*}
    \Delta U_z = \sum_{i=1}^{u} p_i [U_z(b_{g_1(i)}) - U_z(b_{h_1(i)})] + \sum_{i=1}^v q_i [U_z(b_{g_2(i)}) - U_z(b_{h_2(i)})]
\end{align*}
For $z$, balls indexed $i \in [1, A)$ (e.g. $b_{g_1(\cdot)}$) are farther from $z$ than those indexed $i \in [A, Z]$ (e.g. $b_{h_1(\cdot)}$ and $b_{h_2(\cdot)}$), which are farther than those indexed $i \in (R, n]$ (e.g. $b_{g_2(\cdot)}$). Therefore, 
\begin{align*}
    U_z(b_{g_1(i)}) - U_z(b_{h_1(i)}) \leq 0\text{ and } U_z(b_{g_2(i)}) - U_z(b_{h_2(i)}) \geq 0
\end{align*}
Combining these, we get that
\begin{align}
    \Delta U_z \geq 0 \iff \frac{\sum_{i=1}^{u}p_i[U_z(b_{g_1(i)}) - U_z(b_{h_1(i)})]}{\sum_{i=1}^{v}q_i[U_z(b_{h_2(i)}) - U_z(b_{g_2(i)})]} \leq 1\label{eqn:utilityz}
\end{align}
since $q_i[U_z(b_{h_2(i)}) - U_z(b_{g_2(i)})] \leq 0$. 

We now have one last step. For any $f(x)$ which is concave and monotonically decreasing, we have that
    \[\frac{\sum\limits_{i=1}^m c_i[f(t_i^2) - f(t_i^1)]}{\sum\limits_{j=1}^n d_j[f(s_j^2) - f(s_j^1)]} \geq \frac{\sum\limits_{i=1}^m c_i[t_i^2 -t_i^1]}{\sum\limits_{j=1}^n d_j[s_j^2 - s_j^1]} \hspace{0.1 in}\text{and}\hspace{0.1 in} \frac{\sum\limits_{i=1}^m c_i[f(s_i^1) - f(s_i^2)]}{\sum\limits_{j=1}^n d_j[f(t_j^1) - f(t_j^2)]} \leq \frac{\sum\limits_{i=1}^m c_i[s_i^1 -s_i^2]}{\sum\limits_{j=1}^n d_j[t_j^1 - t_j^2]}\]
    for $s_j^1 \leq t_i^1$, $s_j^2 \leq t_i^2$, $s_i^1 \leq s_j^2$, $t_i^1 \leq t_i^2$, and $\text{sign}(c_i) = \text{sign}(d_j)$ (as detailed in Appendix \ref{appsec:supporting}). Applying this to (\ref{eqn:utilityx}) and (\ref{eqn:utilityz}), we get
\begin{align*}
    \frac{\sum_{i=1}^{u}p_i[U_x(b_{g_1(i)}) - U_x(b_{h_1(i)})]}{\sum_{i=1}^{v}q_i[U_x(b_{h_2(i)}) - U_x(b_{g_2(i)})]}
    \leq 
    \frac{\sum_{i=1}^{u}p_i[d(x,b_{g_1(i)}) - d(x, b_{h_1(i)})]}{\sum_{i=1}^{v}q_i[d(x, b_{h_2(i)}) - d(x, b_{g_2(i)})]}
    = \frac{\sum_{i=1}^{u}p_i[-d(b_{g_1(i)}, b_{h_1(i)})]}{\sum_{i=1}^{v}q_i[-d(b_{h_2(i)}, b_{g_2(i)})]}
\end{align*}
and
\begin{align*}
    \frac{\sum_{i=1}^{u}p_i[U_z(b_{g_1(i)}) - U_z(b_{h_1(i)})]}{\sum_{i=1}^{v}q_i[U_z(b_{h_2(i)}) - U_z(b_{g_2(i)})]}
    \geq 
    \frac{\sum_{i=1}^{u}p_i[d(z, b_{g_1(i)}) - d(z, b_{h_1(i)})]}{\sum_{i=1}^{v}q_i[d(z, b_{h_2(i)}) - d(z, b_{g_2(i)})]}
    = \frac{\sum_{i=1}^{u}p_i[d(b_{g_1(i)}, b_{h_1(i)})]}{\sum_{i=1}^{v}q_i[d(b_{h_2(i)}, b_{g_2(i)})]}
\end{align*}
Then either
\[\frac{\sum_{i=1}^{u}p_i[d(b_{g_1(i)}, b_{h_1(i)})]}{\sum_{i=1}^{v}q_i[d(b_{h_2(i)}, b_{g_2(i)})]} \geq 1 \hspace{0.1in}\text{ or }\hspace{0.1in} \frac{\sum_{i=1}^{u}p_i[d(b_{g_1(i)}, b_{h_1(i)})]}{\sum_{i=1}^{v}q_i[d(b_{h_2(i)}, b_{g_2(i)})]} \leq 1\]

If the first is true, we can apply (\ref{eqn:utilityx}) to claim that $\Delta U_x \geq 0$. If the second is true, we can apply (\ref{eqn:utilityz}) to claim that $\Delta U_z \geq 0$. Since one of these must be true, we are done.
\end{proof}

\subsection{Part B: Moving balls into the interval $[x, z]$ only decreases $\Delta U_x$ and $\Delta U_z$}
\begin{lemma}\label{lem:move_outside_points}
    Consider urns $P$ and $Q$. From $P$, create a new urn $P'$ by moving all balls left of $x$ to $x$ and all balls right of $z$ to $z$. Similarly, from $Q$, create a new urn $Q'$ by moving all balls left of $x$ to $x$ and all balls right of $z$ to $z$. Let $\Delta p'(b_i) = \prob[b_i \text{ wins in urn $P'$}] - \prob[b_i \text{ wins in urn $Q'$}]$ and $\Delta U'_x = U_x^{P'} - U_x^{Q'}$. Then,
    \[\Delta U_x \geq \Delta U'_x \hspace{0.1 in}\text{and}\hspace{0.1in} \Delta U_z \geq \Delta U'_z\]
\end{lemma}
\begin{proof}
Since the relative positions of the balls have not changed in $P'$ and $Q'$, the change in winning probabilities from $P'$ to $Q'$ are the same, i.e. $\Delta p(b_i) = \Delta p'(b_i)$. Then for $\Delta p(b_i) < 0$, 
\begin{align*}
    \Delta p(b_i)U_{x}(b_i) \geq \Delta p'(b_i)U'_{x}(b_i) &\iff U_{x}(b_i) \leq U'_{x}(b_i)\\
                                                                 &\iff b_i \text{ is moved closer to $x$}
\end{align*}
and for $\Delta p(b_i) > 0$,
\begin{align*}
    \Delta p(b_i)U_{x}(b_i) \geq \Delta p'(b_i)U'_{x}(b_i) &\iff U_{x}(b_i) \geq U'_{x}(b_i)\\
                                                                 &\iff b_i \text{ is moved farther from $x$}
\end{align*}
Likewise, the same statements hold when $x$ is replaced with $z$. Then we can use this to prove our Lemma by splitting it up into three cases:

\paragraph{Case 1: $b_{n/2} \in [x, z]$ for urn $P$} 

When the median ball is located in $[x, z]$, we know from Lemma \ref{lem:delta_add_tokens} that all balls left of $x$ and right of $z$ satisfy $\Delta p(b_i) < 0$. Then moving the balls left of $x$ to $x$ brings them closer to both $x$ and $z$. Similarly, moving the balls right of $z$ to $z$ brings them closer to both $x$ and $z$. Therefore, we must have
\[\Delta U_{x} = \sum_{i=1}^n \Delta p(b_i)U_{x}(b_i) \geq \sum_{i=1}^n \Delta p'(b_i)U'_{x} = \Delta U'_{x}\]
and similarly, $\Delta U_z \geq \Delta U'_z$.

\paragraph{Case 2: $b_{n/2} < x$ for urn $P$}

When the median ball is located to the left of $x$, we will need to first make an intermediate pair of urns. Note that the balls left of $b_{n/2}$ have $\Delta p(b_i) < 0$ so we can move them rightwards to bring them closer to both $x$ and $z$. However, since the balls between $b_{n/2}$ and $x$ have $\Delta p(b_i) > 0$, we need to move them leftwards to bring them farther away from both $x$ and $z$. For the balls right of $z$, $\Delta p(b_i) < 0$, so we can again move them to $z$ which brings them closer to $x$ and $z$. We will denote these movements with the intermediate urns $P''$ and $Q''$. In these urns, all the balls left of $x$ are moved to the position $b_{n/2}$ and all the balls right of $z$ are moved to $z$. By an argument similar to that of Case 1, $\Delta U_{x} \geq \Delta U''_{x}$ and $\Delta U_z \geq \Delta U''_z$. Now, note that
\begin{align*}
    \Delta U''_{x} &= \left[\sum_{i=1}^{k-1} \Delta p''(b_i)\right] U''_{x}(b_{n/2}) + \sum_{i=k}^n \Delta p''(b_i) U''_{x}(b_i)\\
\end{align*}
where $k$ is the index of the leftmost ball that is right of $x$ and $U'', p''$ are the analogous expressions for utility and winning probability in urns $P'', Q''$. But we also know that for $k < l$,
\begin{align*}
    \sum_{i=1}^{k-1} \Delta p''(b_i) &= \sum\limits_{i=1}^{k-1} \left[\left(\frac{1}{2}\right)^{n-1}\binom{n-1}{i-1} - \left(\frac{1}{2}\right)^{n-4}\binom{n-4}{i-1}\right]\\
                      = \prob&[\text{$\leq k-2$ heads in $n-1$ coin flips}] - \prob[\text{$\leq k-2$ heads in $n-4$ coin flips}] \leq 0
\end{align*}
Then this means that 
\begin{align*}
    \left[\sum_{i=1}^{k-1}\Delta p''(b_i)\right]U''_{x}(b_{n/2}) \geq \left[\sum_{i=1}^{k-1}\Delta p'(b_i)\right]U'_{x}(b_{n/2}) &\iff U''_{x}(b_{n/2} \leq U'_{x}(b_{n/2})\\
    &\iff b_{n/2} \text{ is moved closer to $x$}
\end{align*}
where the same statements hold when $x$ is replaced by $z$. We can now moving the balls at the position of ball $b_{n/2}$ rightwards to $x$ to create the final urns $P'$ and $Q'$. When we do so, we are bringing all the balls closer to both $x$ and $z$, which we just showed must decrease $\Delta U''_x$ and $\Delta U''_z$. Then $\Delta U'_x \leq \Delta U''_x \leq U_x$ and $\Delta U'_z \leq \Delta U''_z \leq U_z$, which concludes our proof for this case.

\paragraph{Case 3: $b_{n/2} > z$ for urn $P$} The proof is symmetric to that of Case 2.
\end{proof}

\section{Supporting Proofs}\label{appsec:supporting}

\begin{theorem}\label{thm:urn_time_app}\cite{Lee2012}
    Let a fixed size urn with $R_0$ red balls out of $n$ total balls have an urn function $f(p)$ for which $\frac{f(p)}{f(1-p)}$ is monotonically decreasing and let $T$ denote the first time when either $R_T = n$ or $R_T = 0$. Then,
    \[\mathbb{E}[T] \leq \frac{1}{q_1} + \sum_{k = 1}^{\lfloor\frac{n}{2}\rfloor - 1} \frac{q_k}{q_{k+1}(q_k - p_k)}\]
    where $p_k = f\left(\frac{n - k}{n}\right)$ and $q_k = f\left(\frac{k}{n}\right)$.
\end{theorem}

\begin{corollary}
    Let a fixed size urn with $R_0$ red balls out of $n$ total balls have an urn function $f(p) = 3p(1-p)^2$ and let $T$ denote the first time when either $R_T = n$ or $R_T = 0$. Then,
    \[\mathbb{E}[\tau] \leq n\ln n + O(n) \]
\end{corollary}
\begin{proof}
    From Theorem \ref{thm:urn_time_app}, we know that
    \begin{align*}
        \mathbb{E}[T] &\leq \frac{1}{q_1} + \sum_{k = 1}^{\lfloor\frac{n}{2}\rfloor - 1} \frac{q_k}{q_{k+1}(q_k - p_k)} = \frac{n^3}{3(n-1)^2} + \frac{n^3}{3}\sum_{k = 1}^{\lfloor\frac{n}{2}\rfloor - 1} \underbrace{\frac{n-k}{(k+1)(n-k-1)^2(n-2k)}}_{(*)}
    \end{align*}
    where
    \begin{align*}
        (*) &= -\frac{1}{n(n-2)(n-k-1)^2} + \frac{4n}{(n+2)(n-2)^2(n-2k)}\\
            &\ \ \ \ +\frac{n+1}{n^2(n+2)(k+1)} -\frac{n^2+n-2}{n^2(n-2)^2(n-k-1)}\\
            &\leq \frac{4n}{(n+2)(n-2)^2(n-2k)}+\frac{n+1}{n^2(n+2)(k+1)}\\
            &= \left(\frac{4}{n^2(n-2k)} + \frac{1}{n^2(k+1)}\right)\left(1 + O\left(\frac{1}{n}\right)\right)
    \end{align*}
    Putting it together, we have
    \begin{align*}
        \mathbb{E}[T] &\leq \frac{n}{3} + o(1) + \frac{n}{3}\left(1 + O\left(\frac{1}{n}\right)\right)\sum_{k = 1}^{\lfloor\frac{n}{2}\rfloor - 1}\left[\frac{4}{n-2k} + \frac{1}{k+1}\right]
    \end{align*}
    Noting that $\sum_{i=1}^k \frac{1}{k} = H_k = \ln n + O(1)$, where $H_k$ is the $k$-th Harmonic number, we have
    \begin{align*}
        \mathbb{E}[T] &\leq O(n) + \frac{n}{3}\left(2 \ln \frac{n}{2} + \ln \frac{n}{2}\right)\\
                         &= n\ln n + O(n)
    \end{align*}
\end{proof}

\begin{lemma}\label{lem:concave_ineq}
    Let $f(x)$ be a concave monotonically decreasing function. Then for $s_j^1, s_j^2, t_i^1, t_i^2 \in \mathbb{R}$ satisfying $s_j^1 \leq t_i^1$, $s_j^2 \leq t_i^2$, $s_i^1 \leq s_j^2$, and $t_i^1 \leq t_i^2$ and $\text{sign}(c_i) = \text{sign}(d_j)$, we have
    \[\frac{\sum\limits_{i=1}^m c_i[f(t_i^2) - f(t_i^1)]}{\sum\limits_{j=1}^n d_j[f(s_j^2) - f(s_j^1)]} \geq \frac{\sum\limits_{i=1}^m c_i[t_i^2 -t_i^1]}{\sum\limits_{j=1}^n d_j[s_j^2 - s_j^1]} \hspace{0.1 in}\text{and}\hspace{0.1 in} \frac{\sum\limits_{i=1}^m c_i[f(s_i^1) - f(s_i^2)]}{\sum\limits_{j=1}^n d_j[f(t_j^1) - f(t_j^2)]} \leq \frac{\sum\limits_{i=1}^m c_i[s_i^1 -s_i^2]}{\sum\limits_{j=1}^n d_j[t_j^1 - t_j^2]}\]
\end{lemma}
\begin{proof}
    Since $f$ is concave, $s_j^1 \leq t_i^1$, and $s_j^2 \leq t_i^2$,
    \[\frac{f(t_i^2) - f(t_i^1)}{t_i^2 - t_i^1} \leq \frac{f(t_i^2) - f(s_j^1)}{t_i^2-s_j^1} \leq \frac{f(s_j^2) - f(s_j^1)}{s_j^2 - s_j^1}\]
    Then by noting that $f$ is monotonically decreasing, $s_j^1 \leq s_j^2$, $t_i^1 \leq t_i^2$, and $\text{sign}(c_i) = \text{sign}(d_j)$, we achieve the statement for a single term on the top and bottom
    \begin{align*}
        \frac{c_i[f(t_i^2) - f(t_i^1)]}{d_j[f(s_j^2) - f(s_j^1)]} \geq \frac{c_i[t_i^2 - t_i^1]}{d_j[s_j^2 - s_j^1]}
    \end{align*}
    If we flip this inequality, we get
    \begin{align*}
        \frac{d_j[f(s_j^2) - f(s_j^1)]}{c_i[f(t_i^2) - f(t_i^1)]} \leq \frac{d_j[s_j^2 - s_j^1]}{c_i[t_i^2 - t_i^1]}
    \end{align*}
    Using this, we can derive,
    \[\frac{\sum\limits_{j=1}^n d_j[f(s_j^2) - f(s_j^1)]}{c_i[f(t_i^2) - f(t_i^1)]} = \sum\limits_{j=1}^n \frac{d_j[f(s_j^2) - f(s_j^1)]}{c_i[f(t_i^2) - f(t_i^1)]} \leq \sum\limits_{j=1}^n \frac{d_j[s_j^2 - s_j^1]}{c_i[t_i^2 - t_i^1]} = \frac{\sum\limits_{j=1}^n d_j[s_j^2 - s_j^1]}{c_i[t_i^2 - t_i^1]}\]
    Then by flipping this inequality, we get
    \[\frac{c_i[f(t_i^2) - f(t_i^1)]}{\sum\limits_{j=1}^n d_j[f(s_j^2) - f(s_j^1)]} \geq \frac{c_i[t_i^2 - t_i^1]}{\sum\limits_{j=1}^n d_j[s_j^2 - s_j^1]}\]
    Finally, we can use this to derive the first part of our final result
    \[\frac{\sum\limits_{i=1}^m c_i[f(t_i^2) - f(t_i^1)]}{\sum\limits_{j=1}^n d_j[f(s_j^2) - f(s_j^1)]} = \sum\limits_{i=1}^m \frac{c_i[f(t_i^2) - f(t_i^1)]}{\sum\limits_{j=1}^n d_j[f(s_j^2) - f(s_j^1)]} \geq \sum\limits_{i=1}^m \frac{c_i[t_i^2 - t_i^1]}{\sum\limits_{j=1}^n d_j[s_j^2 - s_j^1]} = \frac{\sum\limits_{i=1}^m c_i[t_i^2 -t_i^1]}{\sum\limits_{j=1}^n d_j[s_j^2 - s_j^1]}\]
    We can get the second part by simply inverting and multiplying the top and bottom of both sides by $-1$
    \[\frac{\sum\limits_{j=1}^n d_j[f(s_j^2) - f(s_j^1)]}{\sum\limits_{i=1}^m c_i[f(t_i^2) - f(t_i^1)]} \leq \frac{\sum\limits_{j=1}^n d_j[s_j^2 - s_j^1]}{\sum\limits_{i=1}^m c_i[t_i^2 -t_i^1]}\]
\end{proof}
\end{document}